\documentclass[11pt]{article}
\usepackage[margin=1in]{geometry}
\usepackage[latin2]{inputenc}
\usepackage[english]{babel}
\usepackage[cmex10]{amsmath}
\usepackage{amsthm}
\usepackage{amssymb}
\usepackage{xspace}
\usepackage{cite}
\usepackage{tikz}
\usepackage{comment}
\usepackage{caption}
\usepackage{subcaption}
\usepackage{letltxmacro}
\usepackage{comment}
\usetikzlibrary{snakes}

\newtheorem{theorem}{Theorem}[section]
\newtheorem{lemma}[theorem]{Lemma}
\newtheorem{claim}[theorem]{Claim}
\newtheorem{corollary}[theorem]{Corollary}

\theoremstyle{definition}
\newtheorem{definition}[theorem]{Definition}

\newcommand{\cF}{{\mathcal{F}}}
\newcommand{\cg}{G_{\cF_0}}

\newcommand{\pw}{\ensuremath{\mathtt{pw}}\xspace}
\newcommand{\poly}{\mathrm{poly}}
\newcommand{\dist}{\mathrm{dist}}
\newcommand{\Oh}{\ensuremath{\mathcal{O}}}

\def\cqedsymbol{\ifmmode$\lrcorner$\else{\unskip\nobreak\hfil
\penalty50\hskip1em\null\nobreak\hfil$\lrcorner$
\parfillskip=0pt\finalhyphendemerits=0\endgraf}\fi}

\newcommand{\executeiffilenewer}[3]{%
\ifnum\pdfstrcmp{\pdffilemoddate{#1}}%
{\pdffilemoddate{#2}}>0%
{\immediate\write18{#3}}\fi%
} 
\newcommand{\defproblem}[3]{
  \vspace{1.5mm}
\noindent\fbox{
  \begin{minipage}{16cm}
  #1 \\
  {\bf{Input:}} #2  \\
  {\bf{Question:}} #3
  \end{minipage}
  }
  \vspace{1.5mm}
}

\newcommand{\defproblemG}[3]{
  \vspace{2mm}
\noindent\fbox{
  \begin{minipage}{16cm}
  #1 \\
  {\bf{Input:}} #2  \\
  {\bf{Goal:}} #3
  \end{minipage}
  }
  \vspace{2mm}
}

\newcommand{\pathdecomp}{\mathbb{P}\xspace}

\DeclareMathAlphabet{\mathcal}{OMS}{cmsy}{m}{n}

\newcommand{\MC}{{\textsc{Multicolored Clique}}\xspace}
\newcommand{\kSP}{{\textsc{$k$-Set Packing}}\xspace}
\newcommand{\SP}{{\textsc{Set Packing}}\xspace}
\newcommand{\tSP}{{\textsc{$3$-Set Packing}}\xspace}
\newcommand{\tDM}{{\textsc{$3$-Dimensional Matching}}\xspace}
\newcommand{\kDM}{{\textsc{$k$-Dimensional Matching}}\xspace}

\newwrite\tempfile
\immediate\openout\tempfile=references.txt
\newcommand{\writeref}[1]{\immediate\write\tempfile{\unexpanded{#1}}}
\newcommand{\writerefe}[1]{\immediate\write\tempfile{\expandafter{#1}}}
\LetLtxMacro{\oldref}{\ref}
\LetLtxMacro{\oldsection}{\section}
\renewcommand{\ref}[1]{\oldref{#1}\writeref{\oldref{#1} (#1)}\writeref{}}

\title{Improved approximation for $3$-dimensional matching \\via bounded pathwidth local search\thanks{The preliminary version of this paper was presented at the 54th Annual IEEE Symposium on Foundations of Computer Science (FOCS'13). The author is partially supported by
Foundation for Polish Science grant HOMING PLUS/2012-6/2.}}

\author{Marek Cygan
  \thanks{
    Institute of Informatics, University of Warsaw, Poland,
      \texttt{cygan@mimuw.edu.pl}.
  }
}

\begin{document}

\maketitle

\begin{abstract}
One of the most natural optimization problems is the \kSP problem,
where given a family of sets of size at most $k$ one
should select a maximum size subfamily of pairwise disjoint sets.
A special case of $3$-{\sc Set Packing} is the 
well known \tDM problem, which is a maximum hypermatching problem in $3$-uniform tripartite hypergraphs.
Both problems belong to the Karp's list of $21$ NP-complete problems.
The best known polynomial time approximation ratio for \kSP
is $(k+\epsilon)/2$ and goes back to 
the work of Hurkens and Schrijver~[SIDMA'89], which gives $(1.5+\epsilon)$-approximation for \tDM.
Those results are obtained by a simple local search algorithm,
that uses constant size swaps.
%, that is a transformation from one feasible solution to another one, by changing a constant subset of it.

%Halld{\'o}rsson~[SODA'95] has shown that logarithmic
%size swaps lead to an improved approximation ratio,
%however at the cost of quasipolynomial time complexity.
%Therefore a natural question is whether it is possible
%to search the space of $r$-size swaps in $c^{r} \poly(|\cF|)$ time
%for constant $c$ and $k$. 
%We show that this is most likely impossible, i.e. there is no
%such algorithm with $f(r) \poly(|\cF|)$ running time, unless W[1]=FPT,
%where $f$ is some computable function, even for $k=3$.
%Therefore trying to find a $c^r \poly(|\cF|)$ time algorithm which searches the whole
%$r$-size swaps space is not the proper path.

The main result of this paper is a new approach to local search for \kSP
where only a special type of swaps is considered, 
which we call swaps of bounded pathwidth.
We show that for a fixed value of $k$ one can search the space of $r$-size
swaps of constant pathwidth in $c^r \poly(|\cF|)$ time.
Moreover we present an analysis proving that
a local search maximum with respect to $O(\log |\cF|)$-size swaps of
constant pathwidth yields a polynomial time $(k+1+\epsilon)/3$-approximation
algorithm, improving the best known approximation
ratio for \kSP. In particular we improve the approximation
ratio for \tDM from $3/2+\epsilon$ to $4/3+\epsilon$.
\end{abstract}

\section{Introduction}

In the \SP problem, also known as {\sc Hypergraph Matching}, we are given a family $\cF \subseteq 2^U$ 
of subsets of $U$, and the goal is to find a maximum size
subfamily of $\cF$ of pairwise disjoint sets.
\SP is a fundamental problem in combinatorial optimization
with various applications.
A simple reduction from {\sc Independent Set} (where $|\cF|=|V|$)
combined with the hardness result of H{\aa}stad~\cite{haastad}
makes the \SP problem hard to approximate.
When each set of \SP is of size at most $k$
the problem is denoted as \kSP.

\defproblemG{\kSP}{A family $\cF \subseteq 2^U$ of sets of size at most $k$.}
{Find a maximum size subfamily of $\cF$ of pairwise disjoint sets.}

\kSP is a generalization of {\sc Independent Set} in bounded degree graphs,
  as well as $k$-{\sc Dimensional Matching} and is related to plethora of other problems
(see~\cite{lau} for a list of connections between
 \kSP and other combinatorial optimization problems).
In \tDM the universe $U$ is partitioned into $U = X \uplus Y \uplus Z$
and $\cF$ is a subset of $X \times Y \times Z$.

Both \tDM and \SP are 
well studied problems, belonging to Karp's list of 21 NP-hard problems~\cite{karp21}.
A simple greedy algorithm returning any inclusionwise maximal
subfamily of disjoint subsets of $\cF$ gives a $k$-approximation
for \kSP.
One can consider a local search routine, where 
as long as it is possible we remove one set from our current feasible solution
and add two new sets.
We say that such an algorithm uses size $2$ swaps, as two new sets are involved.
It is known that a local search maximum with respect to size $2$ swaps
is a $(k+1)/2$-approximation for \kSP.
If, instead of using swaps of size $2$ we use swaps of 
size $r$ for bigger values of $r$, then the approximation
ratio approaches $k/2$, and that is exactly
the $(k/2+\epsilon)$-approximation algorithm by Hurkens and Schrijver~\cite{hs89}.

Despite significant interest (see Section~\ref{sec:related})
for over $20$ years no improved polynomial time approximation
algorithm was obtained for \kSP, even for the special case of \tDM.
Meanwhile Halld{\'o}rsson~\cite{h95}\footnote{Similar arguments
 were also used by Berman and F{\"u}rer~\cite{berman-furer} for the
   independent set problem in bounded degree graphs.} has shown
that a local search maximum with respect to $\Oh(\log |\cF|)$ size
swaps gives a $(k+2)/3$-approximation, which was recently
improved to $(k+1+\epsilon)/3$~\cite{cgm13}.
Nevertheless enumerating all $\Oh(\log |\cF|)$ size swaps 
takes quasipolynomial time.

\subsection{Our results and techniques}

Based on the work of Halld{\'o}rsson~\cite{h95}
a natural path to transforming a quasipolynomial
time approximation into a polynomial time approximation
would be by designing a~$c^r \poly(|\cF|)$ time algorithm,
where $c$ is a constant.
This is exactly the framework of parameterized complexity\footnote{For further information about parameterized
  complexity we defer the reader to monographs~\cite{fpt1,fpt2,fpt3}.},
where the swap size is a natural parameter.
Unfortunately, we show that this is most likely impossible, i.e. there is no
such algorithm with $f(r) \poly(|\cF|)$ running time, unless W[1]=FPT,
where $f$ is some computable function, even for $k=3$.
We would like to note that W[1]$\neq$FPT is a widely believed assumption,
in particular if W[1]=FPT, then the Exponential Time Hypothesis of~\cite{eth} fails.

\begin{theorem}
\label{thm:w1-intro}
Unless $FPT=W[1]$, there is no $f(r) \poly(|\cF|)$ time algorithm, that given a family $\cF \subseteq 2^U$
of sets of size $3$ and its disjoint subfamily $\cF_0 \subseteq \cF$
either finds a bigger disjoint family $\cF_1 \subseteq \cF$ or verifies
that there is no disjoint family $\cF_1 \subseteq \cF$ such that $|\cF_0 \setminus \cF_1| + |\cF_1 \setminus \cF_0| \le r$,
\end{theorem}

Therefore trying to find a $c^r \poly(|\cF|)$ time algorithm which searches the whole
$r$-size swaps space is not the proper path.
For this reason we introduce a notion of swaps (also called improving sets)
of bounded pathwidth (see Section~\ref{ssec:improving}).
Intuitively a size $r$ swap is of bounded pathwidth,
if the bipartite graph where vertices represent sets that are added and removed,
and edges correspond to  non-empty intersections, is of constant pathwidth.
Using the color-coding technique of Alon et al.~\cite{color-coding}
we show that one can search the space of swaps of size at most $r$ of bounded pathwidth in $c^r \poly(|\cF|)$ time, for a constant $c$.
As the currently best known analysis of local search maximum with respect to logarithmic size swaps of~\cite{cgm13}
relies on swaps of unbounded pathwidth,
we need to develop a different proof strategy, and the core part of
it is contained in Lemma~\ref{lem:improving-tree}.
The algorithm and its analysis complete the main result of this paper,
that is a polynomial time $(k+1+\epsilon)/3$-approximation algorithm,
for any fixed $k$ and~$\epsilon$.

\begin{theorem}
\label{thm:main}
For any $\epsilon > 0$ and any integer $k \ge 3$
there is a polynomial time $(k+1+\epsilon)/3$-approximation algorithm
for \kSP.
\end{theorem}

We believe that the usage of parameterized tools
such as color-coding, pathwidth and W[1]-hardness in the
setting of this work is interesting on its own,
as to the best of our knowledge such tools
have not been previously used in local search based approximation algorithms.

\subsection{Related work}
\label{sec:related}

Even though there was no improvement 
in terms of polynomial time approximation of \kSP (and \tDM)
since the work of Hurkens and Schrijver~\cite{hs89},
both problems are well studied.

One can also consider weighted variant of \kSP,
where we want to select a maximum weight disjoint subfamily of $\cF$.
Arkin and Hassin~\cite{arkin-hassin}
gave a $(k-1+\epsilon)$-approximation algorithm,
Chandra and Halld{\'o}rsson~\cite{chandra-halldorson}
improved it to $(2k+2+\epsilon)/3$-approximation,
later improved by Berman~\cite{berman} to $(k+1+\epsilon)/2$-approximation.
All the mentioned results are based on local search.

Also for the standard (unweighted) \kSP problem
Chan and Lau~\cite{lau} presented a strengthened
LP relaxation, which has integrality gap $(k+1)/2$.

On the other hand, Hazan et al~\cite{hazan}
have shown that \kSP is hard to approximate within a factor of $\Oh(k/\log k)$. 
Concerning small values of $k$, Berman and Karpinski~\cite{berman-karpinski}
obtained a $98/97-\epsilon$ hardness for \tDM, while
Hazan et al.~\cite{hazan2} obtained $54/53-\epsilon$, $30/29-\epsilon$, and $23/22-\epsilon$ hardness
for $4$, $5$ and $6$-{\sc Dimensional Matching} respectively (note that a hardness result for \kDM
directly gives a hardness for \kSP).

Recently Sviridenko and Ward~\cite{sviridenko-ward} have
independently obtained a $(k+2)/3$-approximation algorithm for \kSP.
They observed that the analysis of Halld{\'o}rsson~\cite{h95} 
can be combined with a clever application of the color coding
technique. However to the best of our understanding 
it is not possible to obtain $(k+1+\epsilon)/3$-approximation
for \kSP using the tools of~\cite{sviridenko-ward},
and in particular Sviridenko and Ward do not improve the
approximation ratio for $\tDM$.
The main difference between this article and~\cite{sviridenko-ward}
is in handling sets of the optimum solution, that
intersect exactly one set in a local maximum.

\subsection{Organisation}

We start with preliminaries in Section~\ref{sec:preliminaries},
where we recall standard graph notation together
with the definition of pathwidth and path decompositions.

Section~\ref{sec:main} contains the main result of this paper,
that is the $(k+1+\epsilon)/3$-approximation for \kSP.
First, we introduce the notion of improving set of bounded
pathwidth in Section~\ref{ssec:improving}.
In Section~\ref{ssec:algorithm} we apply the color coding 
technique to obtain a polynomial time algorithm
searching an improving set of logarithmic size
of bounded pathwidth.
In Section~\ref{ssec:analysis} we analyse a local search
maximum with respect to bounded pathwidth improving sets of
logarithmic size.
The heart of our analysis is contained in 
an abstract combinatorial Lemma~\ref{lem:improving-tree}
which is later applied in the proof of Lemma~\ref{lem:analysis}.

The proof of Theorem~\ref{thm:w1-intro} is given
in Section~\ref{sec:hardness}.
Finally, in Section~\ref{sec:conclusions} we conclude
with potential future research directions.

\section{Preliminaries}
\label{sec:preliminaries}

We use standard graph notation. For an undirected graph $G$ by $V(G)$ and $E(G)$
we denote the set of its vertices and edges respectively. By $N_G(v) = \{u : uv \in E(G)\}$
we denote the open neighborhood of a vertex $v$, while the closed neighborhood is defined as $N_G[v] = N_G(v) \cup \{v\}$.
Similarly, for a subset of vertices $X$ we have $N_G[X] = \bigcup_{v \in X} N_G[v]$ and $N_G(X) = N_G[X] \setminus X$.

By a disjoint family of sets we denote a family, where each pair of sets is pairwise disjoint.
For a positive integer $r$ by $[r]$ we denote the set $\{1,\ldots,r\}$.

\paragraph{Pathwidth and path decompositions} A \emph{path decomposition} of a graph~$G=(V,E)$ is a sequence~$\pathdecomp=(B_i)_{i=1}^q$, where each
set $B_i$ is a subset of vertices~$B_i\subseteq V$ (called a \emph{bag}) such that~$\bigcup_{1 \le i \le q} B_i=V$ and the following properties hold:
\begin{itemize}
\item[(i)] For each edge~$uv\in E(G)$ there is a bag $B_i$ in~$\pathdecomp$ such that~$u,v\in B_i$.
\item[(ii)] If~$v\in B_i\cap B_j$ then~$v\in B_{\ell}$ for each $\min(i,j) \le \ell \le \max(i,j)$.
\end{itemize}

The \emph{width} of~$\pathdecomp$ is the size of the largest bag minus one, and the {\em pathwidth} of a graph~$G$ is the minimum width over all possible path decompositions of~$G$. Since our focus here is on path decompositions we only mention in passing that the related notion of \emph{treewidth} can be defined similarly, except for letting the bags of the decomposition form a tree instead of a path.

In order to make the description easier to follow,
it is common to use path decompositions that adhere to some simplifying properties.
The most commonly used notion is that of a nice path decompositions, introduced by Kloks~\cite{Kloks94}; the main idea is that adjacent nodes can be assumed to have bags differing by at most one vertex. 

\begin{definition}[{\bf nice path decomposition}] \label{def:nicepathdecomp}
A \emph{nice path decomposition} is a path decomposition $\pathdecomp=(B_i)_{i=1}^q$, where each bag is of one of the following types:
\begin{itemize}
\item \textbf{First (leftmost) bag}: the bag $B_1$ is empty,~$B_1=\emptyset$.
\item \textbf{Introduce bag}: an internal bag~$B_i$ of $\pathdecomp$ with predecessor~$B_{i-1}$ such that~$B_i = B_{i-1} \cup \{v\}$ for some $v \notin B_{i-1}$. 
This bag is said to \emph{introduce} $v$.
\item \textbf{Forget bag}: an internal bag~$B_i$ of $\pathdecomp$ with predecessor~$B_{i-1}$ for which $B_i = B_{i-1} \setminus \{v\}$ for some $v \in B_{i-1}$. This bag is said to \emph{forget} $v$.
\item \textbf{Last (rightmost) bag}: the bag associated with the largest index, i.e. $q$, is empty,~$B_q=\emptyset$.
\end{itemize}
\end{definition}

It is easy to verify that any given path decomposition can be transformed in polynomial time into a nice path decomposition without increasing its width.

\section{Local search algorithm}
\label{sec:main}

In this section we present the main result of the paper,
i.e. the $(k+1+\epsilon)/3$-approximation algorithm for \kSP,
proving Theorem~\ref{thm:main}.
We start with introducing the notion of improving set of bounded
pathwidth in Section~\ref{ssec:improving}.
Next, in Section~\ref{ssec:algorithm} we apply the color coding 
technique to obtain a polynomial time algorithm
searching an improving set of logarithmic size
of bounded pathwidth.
In Section~\ref{ssec:analysis} we analyse a local search
maximum with respect to bounded pathwidth improving sets of
logarithmic size.
The heart of our analysis is contained in 
an abstract combinatorial Lemma~\ref{lem:improving-tree}
which is later applied in the proof of Lemma~\ref{lem:analysis}.

\subsection{Bounded pathwidth improving set}
\label{ssec:improving}

Let us assume that an instance $\cF \subseteq 2^U$ of \kSP is given.
Moreover by $\cF_0 \subseteq \cF$ we denote some disjoint subfamily of $\cF$,
which we can think of as a current feasible solution of a local search algorithm.
In what follows we define a {\em conflict graph}, which
is a bipartite undirected graph with two independent sets of vertices being $\cF_0$ and $\cF \setminus \cF_0$,
where an edge reflects non-empty intersection.

\begin{definition}[{\bf conflict graph}]
For a disjoint family $\cF_0 \subseteq \cF$ by a {\em conflict graph} $G_{\cF_0}$ we denote
an undirected bipartite graph with vertex set $\cF$ and edge set $\{S_1S_2 : S_1 \in \cF_0, S_2 \in (\cF \setminus \cF_0), S_1 \cap S_2 \neq \emptyset\}$.
\end{definition}

Next, we define an {\em improving set} $X \subseteq \cF \setminus \cF_0$, which
can be used to increase the cardinality of $\cF_0$, and then we introduce 
a notion of an {\em improving set of bounded pathwidth}, which will be crucial in 
both the algorithm and the analysis of its approximation ratio.

\begin{definition}[{\bf improving set}]
\label{def:improving}
For a disjoint family $\cF_0 \subseteq \cF$ a set $X \subseteq \cF \setminus \cF_0$
is called an {\em improving set}, if the following conditions hold:
\begin{itemize}
  \item all sets of $X$ are pairwise disjoint,
  \item $|N_{G_{\cF_0}}(X)| < |X|$, i.e. the number of sets of $\cF_0$ having a common
  element with at least one set of $X$ is strictly smaller than $|X|$.
\end{itemize}
\end{definition}

Observe, that if we have an improving set $X$, then $(\cF_0 \setminus N_{\cg}(X)) \cup X$
is a disjoint subfamily of $\cF$ of size greater than $|\cF_0|$, hence the name improving set.

\begin{definition}[{\bf improving set of bounded pathwidth}]
An improving set $X$ with respect to $\cF_0 \subseteq \cF$ has {\em pathwidth at most $\pw$},
if the subgraph of the conflict graph $G_{\cF_0}$ induced by $N_{\cg}[X]$ has pathwidth at most $\pw$.
\end{definition}

\subsection{Algorithm}
\label{ssec:algorithm}

To find an improving set of bounded pathwidth we use
the color coding technique of Alon et al.~\cite{color-coding},
which is by now a well-established tool in parameterized complexity
used for finding a set consisting of disjoint objects.
We use two random colorings $c_{\cF_0} : \cF_0 \to [r-1]$, $c_U : U \to [rk]$,
where $c_U$ ensures that the sets of $X$ are disjoint,
while $c_{\cF_0}$ is used not to consider the same set of $\cF_0$ twice.

\begin{lemma}
\label{lem:search}
There is an algorithm, that given a disjoint family $\cF_0 \subseteq \cF$,
and two coloring functions $c_{\cF_0} : \cF_0 \to [r-1]$, $c_U : U \to [rk]$
in $2^{\Oh(rk)}|\cF|^{\Oh(\pw)}$ time determines, whether there exists an improving 
set $X \subseteq \cF \setminus \cF_0$ of size at most $r$
of pathwidth at most $\pw$, such that $c_{\cF_0}$ is injective on $N_{G_{\cF_0}}(X)$ 
and $c_U$ is injective on $\bigcup_{S \in X} S$.
\end{lemma}

\begin{proof}
For the sake of notation by adding dummy distinct elements we ensure 
that each set of $\cF$ has size exactly~$k$.
Define an auxiliary directed graph $D=(V_D,A_{forget} \cup A_{introduce})$,
where each vertex is characterized by a subset of set colors $[r-1]$,
a subset of element colors $[rk]$, and a subset of $\cF$ of size at most $\pw+1$,
i.e.  
\begin{align*}
V_H = \{v(C_{\cF_0}, C_{U}, B) & : C_{\cF_0} \subseteq [r-1], C_U \subseteq [rk],\\
     & \quad B \subseteq \cF, |B| \le \pw+1\}\,.
\end{align*}
Note that this graph has $\Oh(2^{r(k+1)}|\cF|^{\pw+1})$ vertices.

Since there will be no possibility of confusion, to make the proof easier
to follow by $N[X]$ for $X \subseteq \cF$ we denote $N_{\cg}[X]$,
i.e. we omit the subscript $\cg$.
The idea behind the construction is that each vertex of $V_H$ describes
a potential prefix of a sequence of bags in a path decomposition of $N[X]$
for some $X \subseteq \cF \setminus \cF_0$.
The set $B$ encodes the set of vertices of $N[X]$ in the current bag
and ensures the bounded pathwidth property.
Instead of storing all the sets of $X$ that have already appeared in the 
sequence of bags, we store only the colors of the elements of $\bigcup_{S \in X} S$ (encoded by~$C_U$),
as it is enough to maintain the disjointness of sets of $X$
and keep track of the cardinality of $X$ - due to the assumption that each 
set of is size exactly~$k$.
Similarly instead of storing all the sets of $N[X]$
that were already introduced, we only store their colors (encoded by~$C_{\cF_0}$).

To the set $A_{introduce}$ we add the following arcs.
For $s = v(C_{\cF_0},C_U,B)  \in V_D$, $S \in \cF$
such that $|B| \le \pw$:
\begin{itemize}
  \item if $S \in \cF \setminus \cF_0$, $c_U(S) \cap C_U = \emptyset$, $c_{\cF_0}$ is injective on $N(S)$ 
  and $c_{\cF_0}(N(S) \setminus B) \cap C_{\cF_0} = \emptyset$,
  then add to $A_{introduce}$ an arc $(s, v(C_{\cF_0}, C_U \cup c_U(S), B \cup \{S\}))$
  \item if $S \in \cF_0$, $c_{\cF_0}(S) \not\in C_{\cF_0}$ and for each $S' \in B \setminus \cF_0$ 
  either $S \in N(S')$, or $c_{\cF_0}(S) \not\in c_{\cF_0}(N(S'))$,
  then add to $A_{introduce}$ an arc $(s, v(C_{\cF_0} \cup \{c_{\cF_0}(S)\}, C_U, B \cup \{S\}))$
\end{itemize}

To the set $A_{forget}$ we add the following arcs.
For $s = v(C_{\cF_0},C_U,B)  \in V_D$, $S \in B$  add to $A_{forget}$
an arc  $(s, v(C_{\cF_0}, C_U, B \setminus \{S\}))$
if one of the following conditions holds:
\begin{itemize}
  \item $S \in \cF_0$,
  \item $S \not\in \cF_0$ and $c_{\cF_0}(N(S)) \subseteq C_{\cF_0}$.
\end{itemize}

\begin{claim}
There exists a path in the graph $D$ from the vertex $v(\emptyset, \emptyset, \emptyset)$
to a vertex $v(C_{\cF_0}, C_U, \emptyset) \in V_D$ for $|C_{\cF_0}| < |C_U| / k$
iff
there exists an improving set $X$ of size at most $r$ of pathwidth at most $\pw$,
such that $c_{\cF_0}$ is injective on $N(X)$ and $c_U$ is injective on $\bigcup_{S \in X} S$.
\end{claim}

\begin{proof}
Assume that there is a path $s_1,\ldots,s_q$ in $H$,
where $s_i=(C_{\cF_0}^i,C_U^i,B_i)$, $s_1=(\emptyset,\emptyset,\emptyset)$ , $|C_{\cF_0}^q| < |C_U^q|/k$ and $B_q=\emptyset$.
Let $X = \bigcup_{1 \le i \le q} B_i \setminus \cF_0$.
By construction of $D$, we have $|X|=|C_U^q| / k \le r$.
By the definition of $A_{introduce}$ and $A_{forget}$
since $B_q=\emptyset$, at the time a vertex $v \in X$ appears
for the first time in some $B_i$ we ensure that all its neighbors in $\cg$
are either in $B_i$ or are colored by $c_{\cF_0}$ with colors not yet in $C_{\cF_0}^i$.
Moreover at the time $v \in X$ is forgotten, i.e. removed from some $B_i$,
we ensure that all of its neighbors in $\cg$ have been already added to bags.
Therefore $N[X] \subseteq \bigcup_{1 \le i \le q} B_i$ and for each 
edge $e$ of $G[N[X]]$ the endpoints of $e$ appear in some bag $B_i$.
Since no set of $\cF_0$ is added twice, due to the coloring $c_{\cF_0}$,
no set of $\cF \setminus \cF_0$ is added twice, due to the coloring $c_U$,
$(B_i \cap N[X])_{i=1}^q$ is a path decomposition of $N[X]$ of width at most $\pw$.
Finally $|N(X)| \le |C^q_{\cF_0}| < |C^q_U|/k = |X|$.
Hence $X$ is an improving set of size at most $r$ and of pathwidth at most $\pw$.

In the other direction, let $X$ be an improving set of size at most $r$ 
such that $c_{\cF_0}$ is injective on $N(X)$, $c_U$ is injective on $\bigcup_{S \in X} S$,
and let $\pathdecomp=(B_i)_{i=1}^q$ be a nice path decomposition of $N[X]$ of width at most $\pw$.
For $1 \le i \le q$ define $s_i \in V_D$ as $s_i = v(c_{\cF_0}(B_i' \cap \cF_0), c_U(\bigcup_{S \in B_i' \setminus \cF_0} S), B_i)$,
where $B_i' = \bigcup_{1 \le j \le i} B_i$.
Observe that $s_1=(\emptyset,\emptyset,\emptyset)$, $s_q=(C_{\cF_0}, C_U, \emptyset)$ for $|C_{\cF_0}| = |N(X)| < |X| = |C_U| / k$
and moreover if $B_{i+1}$ is an introduce bag, then $(s_i,s_{i+1}) \in A_{introduce}$
while if $B_{i+1}$ is a forget bag, then $(s_i,s_{i+1}) \in A_{forget}$.
Consequently there is a path from $s_1$ to $s_q$ in the graph $D$.

\end{proof}

By the above claim it is enough to run a standard graph
search algorithm, to check whether there exists a path
from the vertex $v(\emptyset, \emptyset, \emptyset)$
to $v(C_{\cF_0}, C_U, \emptyset)$ where $|C_{\cF_0}| < |C_U| / k$,
which finishes the proof of Lemma~\ref{lem:search}.
\end{proof}

\begin{theorem}
\label{thm:search}
There is an algorithm, that given a disjoint family $\cF_0 \subseteq \cF$,
in $2^{\Oh(rk)}|\cF|^{\Oh(\pw)}$ time determines, whether there exists an improving 
set $X \subseteq \cF \setminus \cF_0$ of size at most $r$
of pathwidth at most $\pw$.
\end{theorem}

\begin{proof}
Observe, that if we take $c_{\cF_0} : \cF_0 \to [r-1]$
where the color of each set is chosen uniformly and independently at random,
then for an improving set $X$ of size at most $r$ 
the function $c_{\cF_0}$ is injective on $N_{G_{\cF_0}}(X)$
with probability at least $$(r-1)!/(r-1)^{r-1} \ge ((r-1)/e)^{r-1} / (r-1)^{r-1} = e^{-(r-1)}\,.$$
Similarly, if we assign a color of $[rk]$ to each element of $U$,
then with probability at least $e^{-rk}$ the function $c_U : U \to [rk]$ is injective on $\bigcup_{S \in X} S$.
Therefore invoking Lemma~\ref{lem:search} with random colorings $c_{\cF_0}, c_U$ at least $e^{r-1+rk}$ 
times would yield a constant error probability.

To obtain a deterministic algorithm we can use the, by now standard, technique of splitters.  An $(n,a,b)$-splitter is a family $\mathcal{H}$ of functions $[n] \to [b]$, such that for any $W \subseteq [n]$ of size at most $a$ there exists $f \in \mathcal{H}$ that is injective on $W$.  What we need is a small family of $(n,a,a)$-splitters.

\begin{theorem}[ \hskip -0.25cm \cite{opt-splitters}]
\label{thm:splitters}
There exists an $(n,a,a)$-splitter of size $e^a a^{\Oh(\log a)} \log n$
that can be constructed in $\Oh(e^a a^{\Oh(\log a)} n \log n)$ time.
\end{theorem}

Therefore instead of using random colorings $c_{\cF_0}$, $c_U$ we can use Theorem~\ref{thm:splitters}
to construct $(|\cF_0|,r-1,r-1)$ and $(|U|,rk,rk)$ splitters, leading to a deterministic algorithm,
which finishes the proof of Theorem~\ref{thm:search}.
\end{proof}

\subsection{Analysis}
\label{ssec:analysis}

In this subsection we analyze a local search maximum, with respect
to logarithmic size improving sets of constant pathwidth.
It is well known that an undirected graph of average degree
at least $2+\epsilon$ contains a cycle of length at most $c_\epsilon \log n$,
   where the constant $c_{\epsilon}$ depends on $\epsilon$.
This observation was the base for the quasipolynomial time algorithms of~\cite{h95,cgm13}.
Here, however we need to generalize this result extensively,
as the analysis of~\cite{cgm13} relies on improving sets of unbounded pathwidth.

Throughout this subsection we operate on multigraphs, as there might be several
parallel edges in a graph, however there will be no self-loops.

\begin{lemma}
\label{lem:improving-tree}
Let $H=(V,E)$ be an $n$-vertex undirected multigraph of minimum degree at least $3$.
Assume that each edge $e \in E$ is associated with a subset of an alphabet $w_e \subseteq \Sigma$ of size at most $\gamma$, where $\gamma \ge 1$.
If each element $c \in \Sigma$ appears in at most $\gamma$ sets $w_e$, i.e. $\forall_{c\in \Sigma}~|\{e : e \in E, c \in w_e\}| \le \gamma$, then
there exists a tree $T_0 = (V_0, E_0)$, which is a subgraph of $H$, and a vertex $r_0 \in V_0$, such that:
\begin{itemize}
  \item $|V_0| \le 4(\log_{3/2} n + 2)$;
  \item there exist two edges $e_1, e_2 \in E \setminus E_0, e_1 \neq e_2$ which have both endpoints in $V_0$;
  \item $T_0$ is a tree with at most $4$ leaves; 
  \item for each pair of edges $e_1,e_2 \in E_0$ such that $w_{e_1} \cap w_{e_2} \neq \emptyset$
  we have $|\dist_{T_0}(r_0,e_1) - \dist_{T_0}(r_0,e_2)| \le \beta$, where $\beta = \lceil \log_{3/2} (12\gamma^2) \rceil$,
  and $\dist_{T_0}(r_0,uv) = \min(\dist_{T_0}(r_0,u), \dist_{T_0}(r_0,v))$.
\end{itemize}
\end{lemma}

\begin{proof}
First we deal with some corner cases. 
\begin{itemize}
\item[(i)] If in $H$ there are three parallel edges $e_a,e_b,e_c$ between vertices $u$ and $v$,
then as $T_0$ we take $(\{u,v\},\{e_a\})$ and set $e_1 = e_b$, $e_2 = e_c$.
\item[(ii)] If in $H$ there are three vertices $u,v,w$, two parallel edges $e_a, e_b$ between $u$ and $v$
as well as two parallel edges $e_c,e_d$ between $v$ and $w$, 
than as $T_0$ we take $(\{u,v,w\}, \{e_a,e_c\})$ and set $e_1=e_b$, $e_2 = e_d$.
\item[(iii)] In the last corner case let us assume that for each vertex $v$ of $H$
there are some two parallel edges $e_a, e_b \in E(H)$ incident to $v$.
Let $uv \in E(H)$ be any edge of $H$ for which there is no parallel edge in $H$ -
such an edge exists, as otherwise $(i)$ or $(ii)$ would hold.
Let $u'$ be a vertex such that in $H$ there are two parallel edges $e_a,e_b$ between $u$ and $u'$,
similarly let $v'$ be a vertex such that in $H$ there are two parallel edges $e_c,e_d$ between $v$ and $v'$.
Observe that $u'\neq v'$ as otherwise case (ii) would hold.
In that case $T_0=(\{u,u',v,v'\},\{e_a,uv,e_c\})$, $e_1 = e_b$ and $e_2 = e_d$.
\end{itemize}
Assuming that none of $(i)$, $(ii)$, $(iii)$ holds,
there is a vertex $r$ in $H$, which is adjacent to at least three distinct vertices $v_1,v_2,v_3$.

We are going to construct a sequence of logarithmic number of trees $T_1, T_2, \ldots$ rooted at $r$, which are subgraphs of $H$ satisfying two invariants:
\begin{itemize}
  \item {\bf (exponential growth)} for any $1 \le j \le i$ the number of vertices in $T_i$
  at distance exactly $j$ from $r$ is exactly $\lfloor 2(3/2)^j \rfloor$,
  and there are no vertices at distance more than $i$,
  \item{\bf ($\Sigma$-nearness)} for any two edges $e_1,e_2$ of $T_i$ if $w_{e_1} \cap w_{e_2} \neq \emptyset$,
  then $|\dist_{T_i}(r,e_1) - \dist_{T_i}(r,e_2)| \le \beta$.
\end{itemize}
We will show, that having constructed a tree $T_i$ for some $i \ge 1$ we
can either construct a tree $T_{i+1}$ satisfying the two invariants,
or find a tree $T_0$ with edges $e_1,e_2$ required by the claim of the lemma.

Let $T_1=(\{r,v_1,v_2,v_3\},\{rv_1,rv_2,rv_3\})$ and note that it satisfies the two invariants.
Assume, that $T_i$ (for some $i \ge 1$) was the most recently constructed tree,
and we want to construct $T_{i+1}$.
Let $V'$ be the vertices of $T_i$ at distance exactly $i$ from the root $r$.
By the exponential growth invariant we have $|V'| = \lfloor 2(3/2)^i \rfloor$.
Let $E' \subseteq E$ be the set of edges of $H$ incident to $V'$, but not contained in $E(T_i)$.
As each vertex in $H$ is of degree at least three, we have
\begin{align}
\label{eq0}
|E'| \ge 2|V'| \ge 2\lfloor 2(3/2)^j \rfloor\,.
\end{align}
Let $$E_{banned} = \{e \in E': \exists_{e' \in E(T_{i-\beta})} w_e \cap w_{e'} \neq \emptyset\}\,,$$
i.e. the set of edges having a non-empty intersection with $w_{e'}$, where $e'$ is not contained
in the last $\beta$ levels of $T_i$.  Observe that for $i \le \beta$ the set $E_{banned}$ is empty.
When extending the tree $T_i$ to maintain the $\Sigma$-nearness invariant,
we use only edges of $E' \setminus E_{banned}$.

Let $V'' = \bigcup_{uv \in E' \setminus E_{banned}} \{u,v\} \setminus V(T_i)$.
We consider two cases: either $|V''| \ge \lfloor 2(3/2)^{i+1} \rfloor$
or not. In the former case we will show that one can construct a tree $T_{i+1}$ satisfying both exponential growth and $\Sigma$-nearness invariants.
In the latter case we will show that the required tree $T_0$ and edges $e_1,e_2$ exist.

If $|V''| \ge \lfloor 2(3/2)^{i+1} \rfloor$,
then we select exactly $\lfloor 2(3/2)^{i+1} \rfloor$ vertices out of $V''$
and extend the tree $T_i$ to $T_{i+1}$ by adding one more layer of vertices (at distance $i+1$ from $r$),
connected to vertices of $V'$ by edges of $E' \setminus E_{banned}$.
Clearly the exponential growth invariant is satisfied for $T_{i+1}$. 
Furthermore, since $T_i$ satisfied the $\Sigma$-nearness invariant
and by the definition of $E_{banned}$ the tree $T_{i+1}$ also satisfies the $\Sigma$-nearness invariant.

In the remaining part of the proof we assume 
\begin{align}
\label{eq1}
|V''| < \lfloor 2(3/2)^{i+1} \rfloor
\end{align}
and show the required tree $T_0$ with edges $e_1,e_2$.
If at least two edges of $E'$ have both endpoints in $V(T_i)$,
denote those edges $uv,u'v' \in E'$, then as $T_0$ we take the subtree of $T_i$ induced
by vertices on the paths between $\{u,v,u',v'\}$ and their least common ancestor $r_0$
and set $e_1=uv$, $e_2=u'v'$ (see Figure~\ref{fig3}).
Therefore let $E'' \subseteq E'$ be the subset of edges having
exactly one endpoint in $V(T_i)$ (that is in $V'$).
By~(\ref{eq0}) we infer that
\begin{align}
\label{eq2}
|E''| \ge |E'|-1 \ge 2|V'|-1\,.
\end{align}

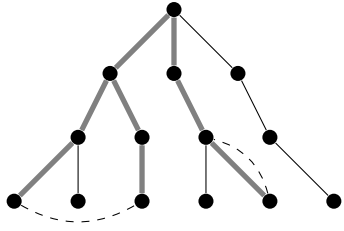
\begin{figure}
\begin{center}
\begin{tikzpicture}[scale=0.85]
  \tikzstyle{vertex}=[circle,fill=black,minimum size=0.20cm,inner sep=0pt]
  \tikzstyle{vertex2}=[circle,draw=black,fill=gray!50,minimum size=0.20cm,inner sep=0pt]
 	\tikzstyle{vertex3}=[circle,draw=black,fill=white,minimum size=0.20cm,inner sep=0pt]
 
  \tikzstyle{terminal}=[rectangle,draw=black,fill=white,minimum size=0.2cm,inner sep=0pt]

\node[vertex] (v0) at (0,0){};
\node[vertex] (v1) at (-1,-1){};
\node[vertex] (v2) at (0,-1){};
\node[vertex] (v3) at (1,-1){};

\node[vertex] (v4) at (-1.5,-2){};
\node[vertex] (v5) at (-0.5,-2){};
\node[vertex] (v6) at (0.5,-2){};
\node[vertex] (v7) at (1.5,-2){};

\node[vertex] (v8) at (-2.5,-3){};
\node[vertex] (v9) at (-1.5,-3){};
\node[vertex] (v10) at (-0.5,-3){};
\node[vertex] (v11) at (0.5,-3){};
\node[vertex] (v12) at (1.5,-3){};
\node[vertex] (v13) at (2.5,-3){};

\draw (v1) -- (v0) -- (v2);
\draw (v0) -- (v3);

\draw (v4) -- (v1) -- (v5);
\draw (v6) -- (v2);
\draw (v7) -- (v3);

\draw(v8) -- (v4) -- (v9);
\draw (v10) -- (v5);
\draw (v11) -- (v6) -- (v12);
\draw (v13) -- (v7);

\draw[dashed] (v8) edge[bend right] (v10);
\draw[dashed] (v12) edge[bend right] (v6);

\draw[color=gray,line width=2] (v8) -- (v4) -- (v1) -- (v0);
\draw[color=gray,line width=2] (v10) -- (v5) -- (v1);

\draw[color=gray,line width=2] (v12) -- (v6) -- (v2) -- (v0);

\end{tikzpicture}
\end{center}
\caption{Edges of the tree $T_0$ are gray, while edges $e_1$ and $e_2$ are dashed.}
\label{fig3}
\end{figure}

Let $E'''$ be a maximum size subset of $E''$, such
that no two edges of $E'''$ have a common endpoint in $V \setminus V(T_i)$.
Observe that if $|E'''| \le |E''|-2$, then either:
\begin{itemize}
  \item there exists three edges $e_a,e_b,e_c \in E''$ having a common endpoint in $V \setminus V(T_i)$, or
  \item there exist four edges $e_a,e_b,e_c,e_d \in E''$, such that $e_a,e_b$ have a common endpoint in $V \setminus V(T_i)$
  and $e_c,e_d$ have a common endpoint in $V \setminus V(T_i)$.
\end{itemize}
In both cases we can extend the tree $T_i$ by one or two edges to construct
$T_0$ and set $e_1=e_b$, $e_2=e_c$ (see Figure~\ref{fig4}).

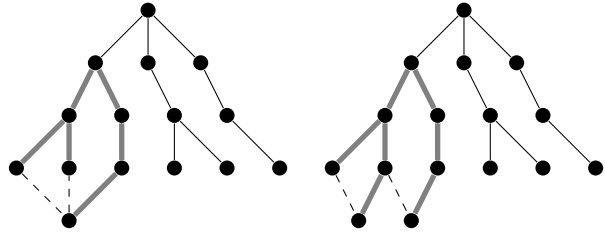
\begin{figure}
\begin{center}
\begin{tikzpicture}[scale=0.70]
  \tikzstyle{vertex}=[circle,fill=black,minimum size=0.20cm,inner sep=0pt]
  \tikzstyle{vertex2}=[circle,draw=black,fill=gray!50,minimum size=0.20cm,inner sep=0pt]
 	\tikzstyle{vertex3}=[circle,draw=black,fill=white,minimum size=0.20cm,inner sep=0pt]
 
  \tikzstyle{terminal}=[rectangle,draw=black,fill=white,minimum size=0.2cm,inner sep=0pt]

\node[vertex] (v0) at (0,0){};
\node[vertex] (v1) at (-1,-1){};
\node[vertex] (v2) at (0,-1){};
\node[vertex] (v3) at (1,-1){};

\node[vertex] (v4) at (-1.5,-2){};
\node[vertex] (v5) at (-0.5,-2){};
\node[vertex] (v6) at (0.5,-2){};
\node[vertex] (v7) at (1.5,-2){};

\node[vertex] (v8) at (-2.5,-3){};
\node[vertex] (v9) at (-1.5,-3){};
\node[vertex] (v10) at (-0.5,-3){};
\node[vertex] (v11) at (0.5,-3){};
\node[vertex] (v12) at (1.5,-3){};
\node[vertex] (v13) at (2.5,-3){};

\draw (v1) -- (v0) -- (v2);
\draw (v0) -- (v3);

\draw (v4) -- (v1) -- (v5);
\draw (v6) -- (v2);
\draw (v7) -- (v3);

\draw(v8) -- (v4) -- (v9);
\draw (v10) -- (v5);
\draw (v11) -- (v6) -- (v12);
\draw (v13) -- (v7);

\node[vertex] (x) at (-1.5,-4){};
\draw[dashed] (v8) -- (x) -- (v9);
\draw (x) -- (v10);

\draw[color=gray,line width=2] (v8) -- (v4) -- (v1);
\draw[color=gray,line width=2] (v9) -- (v4);
\draw[color=gray,line width=2] (x) -- (v10) -- (v5) -- (v1);

\begin{scope}[xshift=6cm]
\node[vertex] (v0) at (0,0){};
\node[vertex] (v1) at (-1,-1){};
\node[vertex] (v2) at (0,-1){};
\node[vertex] (v3) at (1,-1){};

\node[vertex] (v4) at (-1.5,-2){};
\node[vertex] (v5) at (-0.5,-2){};
\node[vertex] (v6) at (0.5,-2){};
\node[vertex] (v7) at (1.5,-2){};

\node[vertex] (v8) at (-2.5,-3){};
\node[vertex] (v9) at (-1.5,-3){};
\node[vertex] (v10) at (-0.5,-3){};
\node[vertex] (v11) at (0.5,-3){};
\node[vertex] (v12) at (1.5,-3){};
\node[vertex] (v13) at (2.5,-3){};

\draw (v1) -- (v0) -- (v2);
\draw (v0) -- (v3);

\draw (v4) -- (v1) -- (v5);
\draw (v6) -- (v2);
\draw (v7) -- (v3);

\draw(v8) -- (v4) -- (v9);
\draw (v10) -- (v5);
\draw (v11) -- (v6) -- (v12);
\draw (v13) -- (v7);

\node[vertex] (x) at (-2,-4){};
\node[vertex] (y) at (-1,-4){};
\draw[dashed] (v8) -- (x) -- (v9) -- (y);

\draw[color=gray,line width=2] (v8) -- (v4) -- (v1);
\draw[color=gray,line width=2] (x) -- (v9) -- (v4);
\draw[color=gray,line width=2] (y) -- (v10) -- (v5) -- (v1);
\end{scope}

\end{tikzpicture}
\end{center}
\caption{Creating the tree $T_0$ assuming $|E'''| \le |E''|-2$. Notation as in Figure~\ref{fig3}.}
\label{fig4}
\end{figure}

Consequently we have $|E'''| \ge |E''|-1$, which together with~(\ref{eq2}) gives
\begin{align}
\label{eq3}
|E'''| \ge 2|V'|-2\,.
\end{align}

In the last part of the proof we use the following claim.
\begin{claim}
\label{claim:bound}
$$|E'''\setminus E_{banned}| \ge \lfloor 2(3/2)^{i+1} \rfloor$$
\end{claim}
\begin{proof}
Recall that if $i \le \beta$, the set $E_{banned}$ is empty.
Hence by inequality~(\ref{eq3}) in that case $|E''' \setminus E_{banned}| = |E'''| \ge 2 \lfloor 2 (3/2)^i \rfloor - 2$.
A direct check shows that for each $1 \le i \le 4$ we have $2 \lfloor 2 (3/2)^i \rfloor - 2 \ge \lfloor 2(3/2)^{i+1} \rfloor$,
which proves the claim in the case $i \le 4$.

When $4 < i \le \beta$ we have 
\begin{align*}
|E'''\setminus E_{banned}| & \ge 2 \lfloor 2 (3/2)^i \rfloor - 2 \\
    & \ge 2(2(3/2)^i-1)- 2 \ge 2(3/2)^{i+1}\,.
\end{align*}

Finally for $i > \beta$ we upper bound the size of $E_{banned}$
\begin{align*}
%|E_{banned}| & \le \sum_{j=1}^{i-\beta} \gamma^2|V(T_j)| \le 2\gamma^2\sum_{j=1}^{i-\beta} (3/2)^j \\
|E_{banned}| & \le \sum_{j=1}^{i-\beta} \gamma^22(3/2)^j \le 3\gamma^2\sum_{j=0}^{i-\beta-1} (3/2)^j \\
 &  \le 6\gamma^2 ((3/2)^{i-\beta}-1)  \le \frac{(3/2)^i}{2} - 6\,.
\end{align*}
The first inequality follows from the assumption, that each set $w_e$ is of size at most $\gamma$ and each element of $\Sigma$
is contained in at most $\gamma$ sets $w_e$, hence each of $T_i$ contributes at most $\gamma^2$ edges to $E_{banned}$.
The last inequality follows from the choice of $\beta$ and the assumption $\gamma \ge 1$.
Therefore 
\begin{align*}
|E'''\setminus E_{banned}| & \ge |E'''| - |E_{banned}| \\
    & \ge 2 \lfloor 2 (3/2)^i \rfloor - 2 - (\frac{(3/2)^i}{2} - 6) \\
    & \ge 2(3/2)^{i+1}\,.
\end{align*}
\end{proof}

Observe that by the definition of $E'''$ we have $|V''| \ge |E'''\setminus E_{banned}|$,
but then Claim~\ref{claim:bound} contradicts inequality~(\ref{eq1}).
\end{proof}

\begin{corollary}
\label{cor:decomposition}
Let $H=(V,E)$ be an undirected multigraph with $n$ vertices
and of minimum degree at least $3$.
Assume that each edge $e \in V$ is associated with a subset of an alphabet $w_e \subseteq \Sigma$ 
of size at most $\gamma$, for some $\gamma \ge 1$, such that each element of $\Sigma$ appears in at most $\gamma$ sets $w_e$.
There exists a subgraph $H_0 = (V_0, E_0)$ of $H$, and a path decomposition $(B_i)_{i=1}^{q}$ of 
$H_0$ of width at most $4(\beta+3)$, where $\beta=\lceil \log_{3/2} (12\gamma^2) \rceil$ and 
\begin{itemize}
  \item[(a)] $|E_0| = |V_0| + 1$,
  \item[(b)] $|V_0| \le 4(\log_{3/2} n + 2)$,
  \item[(c)]  for each pair of edges $e_1,e_2 \in E_0$, such that $w_{e_1} \cap w_{e_2} \neq \emptyset$
  there exists a bag $B_i$ for some $1 \le i \le q$, such that all of the endpoints of both $e_1$ and $e_2$ are contained
  in $B_i$,
  \item[(d)] for each edge $uv \in E_0$ the set of indices $\{ i : u,v \in B_i\}$
  is an interval.
\end{itemize}
\end{corollary}

\begin{proof}
First, we use Lemma~\ref{lem:improving-tree} to obtain $T_0 = (V_0, E_0)$, $r_0 \in V_0$, such that 
$|V_0| \le 4(\log_{3/2} n + 2)$, where for each pair of edges $e_1,e_2 \in E_0$ such that $w_{e_1} \cap w_{e_2} \neq \emptyset$
we have $|\dist_{T_0}(r_0,e_1) - \dist_{T_0}(r_0,e_2)| \le \beta$.
Let $e_1,e_2 \in E\setminus E_0$ be two edges with both endpoints in $V_0$.
Define $H_0 = (V_0, E_0 \cup \{e_1,e_2\})$, clearly $H_0$ is a subgraph of $H$
and the number of edges is the number of vertices plus one.
Therefore properties $(a)$ and $(b)$ are satisfied and
it remains to show that there exists a path decomposition of $H_0$ of width at most $4(\beta+3)$,
satisfying $(c)$ and $(d)$.

Let $D_i$ be the set of vertices of $V_0$ at distance exactly $i$ from 
$r_0$ in $T_0$.
Consider a sequence $(B_i)_{i=0}^q$, where $q={4(\log_{3/2} n + 2)}$,
and $B_i = \bigcup_{\max(0,i-\beta-1) \le j \le i} D_i \cup e_1 \cup e_2$.
It is straightforward to check that this is in fact a path decomposition of $H_0$,
and since $T_0$ has at most $4$ leaves, this implies that the size of each $D_i$
is upper bounded by $4$, and hence the path decomposition is of width
at most~$4(\beta+3)$.

Observe that property $(c)$ required by the corollary follows from
the last property of Lemma~\ref{lem:improving-tree},
because all of the endpoints of edges $e_1,e_2 \in E_0$, such that $w_{e_1} \cap w_{e_2} \neq \emptyset$,
are contained in $B_{\max(\dist_{T_0}(r_0,e_1)+1, \dist_{T_0}(r_0,e_2)+1)}$.
To prove property $(d)$ let $e = uv$ be an arbitrary edge of $E_0$
and define $I_u = \{i : u \in B_i\}$ and $I_v = \{i : v \in B_i\}$.
As we already know that $(B_i)_{i=0}^q$ is a path decomposition
it follows that both sets $I_u$, $I_v$ form an interval,
hence $I_u \cap I_v$ is also an interval, which proves $(d)$.
\end{proof}

\begin{lemma}
\label{lem:analysis}
Fix an arbitrary $\epsilon > 0$. 
There are constants $c_1, c_2$ (depending on $k$ and $\epsilon$),
such that for any disjoint family $\cF_0 \subseteq \cF$,
for which there is no improving set of size at most $c_1 \log n$ of pathwidth at most $c_2$
we have $|OPT| \le ((k+1)/3+\epsilon) |\cF_0|$,
where $OPT \subseteq \cF$ is a maximum size disjoint subfamily of $\cF$.
\end{lemma}

\begin{proof}
Let $C = \cF_0 \cap OPT$ and denote $A_0 = \cF_0 \setminus C$, $B_0 = OPT \setminus C$.
Let $G_0$ be the subgraph of $G_{\cF_0}$ induced by $A_0 \cup B_0$.
We are going to construct a sequence of at most $1/\epsilon$ subgraphs of $G_0$,
namely $G_i = G_0[A_i \cup B_i]$ for $i \ge 1$, where $A_i \subseteq A_0$, $B_i \subseteq B_0$, 
satisfying two invariants:
\begin{itemize}
 \item[(a)] in $G_i$ there is no subset $X \subseteq B_i$ of size at most $2(k+1)^{1/\epsilon - i}$, such that $|N_{G_i}(X)| < |X|$,
 \item[(b)] $|A_0 \setminus A_i| = |B_0 \setminus B_i|$.
\end{itemize}
Observe $G_0$ trivially satisfies $(b)$
and in order to make $G_0$ satisfy $(a)$ 
it is enough to set $c_1$ and $c_2$ so that
\begin{align*}
c_1 & \ge 2(k+1)^{1/\epsilon}\,, \\
c_2 & \ge 4(k+1)^{1/\epsilon}\,,
\end{align*}
as there is no improving set of size at most $2(k+1)^{1/\epsilon}$
and pathwidth of an improving set of size $x$ is at most $2x$.
Consider subsequent values of $i$ starting from $0$.
Split the vertices of $B_i$ into groups $B_i^1, B_i^2, B_i^3$, 
consisting of vertices of $B_i$ of degree exactly one, exactly two and at least three in $G_i$, respectively.
Observe that because of $(a)$ there is no isolated vertex of $B_i$ in $G_i$
and moreover no two vertices of $B_i^1$ have a common neighbour in $G_i$.
Consider the following two cases:
\begin{itemize}
  \item $|B_i^1| \ge \epsilon |OPT|$:
  in this case we construct a graph $G_{i+1}=G_0[A_{i+1} \cup B_{i+1}]$, where
  $A_{i+1} = A_i \setminus N_{G_i}(B_i^1)$ and $B_{i+1} = B_i^2 \cup B_i^3 = B_i \setminus B_i^1$. 
  The invariant $(a)$ is satisfied, as any set $X \subseteq B_{i+1}$ of size at most $2(k+1)^{1/\epsilon-i-1}$ 
  such that $|N_{G_{i+1}}(X)| < |X|$ would imply existence of 
  a set $X' = X \cup (N_{G_i}(N_{G_i}(X)) \cap B_i^1)$ of size at most $(k+1) \cdot |X| \le 2(k+1)^{1/\epsilon -i}$,
  such that $|N_{G_{i}}(X')| < |X'|$ (see Figure~\ref{fig1}).
  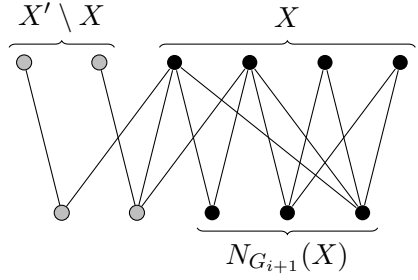
\begin{figure}
  \begin{center}
  \begin{tikzpicture}[scale=1]
  \tikzstyle{vertex}=[circle,fill=black,minimum size=0.20cm,inner sep=0pt]
  \tikzstyle{vertex2}=[circle,draw=black,fill=gray!50,minimum size=0.20cm,inner sep=0pt]
  \tikzstyle{terminal}=[rectangle,draw=black,fill=white,minimum size=0.2cm,inner sep=0pt]

	\foreach \x in {2,3,4,5}
	{
		\node[vertex] (a\x) at (\x,2){};
	}
	\node[vertex2] (a0) at (0,2){};
	\node[vertex2] (a1) at (1,2){};

	\foreach \x in {2,3,4}
	{
		\node[vertex] (b\x) at (\x+0.5,0){};
	}
	\node[vertex2] (b0) at (0.5,0){};
	\node[vertex2] (b1) at (1.5,0){};

	\draw (a1) -- (b1);
	\draw (a0) -- (b0);
	\draw (a2) -- (b2);
	\draw (a2) -- (b4);
	\draw (a2) -- (b0);
	\draw (a2) -- (b1);
	\draw (a3) -- (b1);
	\draw (a3) -- (b4);
	\draw (a3) -- (b3);

\draw (a3) -- (b2);
\draw (a4) -- (b4);
\draw (a4) -- (b3);
\draw (a5) -- (b4);
\draw (a5) -- (b3);

\draw[snake=brace] (1.8,2.2) -- (5.2,2.2);
\draw (3.5,2.6) node {$X$};

\draw[snake=brace] (-0.2,2.2) -- (1.2,2.2);
\draw (0.5,2.6) node {$X' \setminus X$};

\draw[snake=brace] (4.7,-0.2) -- (2.3,-0.2);
\draw (3.5,-0.6) node {$N_{G_{i+1}}(X)$};

\end{tikzpicture}
  \end{center}
  \caption{Lifting an improving set $X$ in $G_{i+1}$ to an improving set $X'$ in $G_{i}$. Gray vertices belong to $G_i$ but not to $G_{i+1}$.}
  \label{fig1}
  \end{figure}
  \item $|B_i^1| < \epsilon |OPT|$:
  We are going to use the following claim, which we prove later.
  \begin{claim}
  \label{claim:new}
  $$|B_i^2| \le (1+\epsilon) |A_i|$$
  \end{claim}

  As each vertex of $A_i$ is of degree at most $k$ in $G_i$,
  the number of edges of $G_i$ is at most $k|A_i|$.
  At the same time the number of edges of $G_i$ is at least 
  $|B_i^1|+2|B_i^2|+3|B_i^3|$, therefore
  \begin{align*}
   |B_i^1| + 2|B_i^2| + 3|B_i^3| &\le k|A_i|\,. 
   \end{align*}

  Note that summing the inequalities:
  \begin{align*}
  |B_i^1| &\le  \epsilon |A_i| \\
  |B_i^1| &\le \epsilon |A_i| \\
  |B_i^2| &\le  (1+\epsilon) |A_i| \\
   |B_i^1| + 2|B_i^2| + 3|B_i^3| &\le k|A_i| 
   \end{align*}
   we obtain

   \begin{align*}
   |B_i| \le ((k+1)/3 + \epsilon)|A_i|\,. 
   \end{align*}
   However $|OPT \setminus B_i| = |C| + |B_0 \setminus B_i| = |C| + |A_0 \setminus A_i| = |\cF_0 \setminus A_i|$, 
   where the second equality follows from invariant $(b)$,
   hence $|OPT| \le ((k+1)/3 + \epsilon)|\cF_0|$.
\end{itemize}
In the second case we have proved the thesis, while the first case can 
appear only $1/\epsilon$ number of times, as in each step we remove at least $\epsilon |OPT|$ vertices from $B_i$.
Therefore to finish the proof of Lemma~\ref{lem:analysis} it suffices to prove Claim~\ref{claim:new}.
\end{proof}

\begin{proof}[Proof of Claim~\ref{claim:new}]
Assume the contrary.
Construct a multigraph $H=(A_i, E_H)$, where $E_H = \{ e_x=uv : x \in B_i^2, N_{G_i}(x) = \{u,v\} \}$.
Set $\Sigma = \cF$ and for each edge $e_x=uv \in E_H$,  set as $w_{e_x}$ the set of
all vertices of $G_0$ at distance at most $2/\epsilon$ from $x$.
Observe that since $G_0$ is of maximum degree at most $k$, we have $|w_{e_x}| \le 2k^{2/\epsilon}$.
For the same reason each vertex of $G_0$ appears in at most $2k^{2/\epsilon}$ sets $w_{e_x}$.

In order to use Corollary~\ref{cor:decomposition} we need to reduce the graph $H$,
in a way ensuring all its vertices are of degree at least $3$.
However we know, that the graph $H$ is of average degree at least $2+2\epsilon$, since $|E_H|/|A_i| = |B_i^2|/|A_i| \ge 1+\epsilon$.
Let $H'=H$. As long as there exist an isolated vertex, or a vertex of degree one in $H'$ remove it.
Note that such a reduction rule does not decrease the average degree of $H'$.
Similarly if $H'$ contains a path $v_0,v_1,\ldots,v_{\ell},v_{\ell+1}$,
where all vertices $v_j$ for $1 \le j \le \ell$ are of degree exactly $2$ and $\ell \ge 1/\epsilon$,
then remove all the vertices $v_j$ for $1 \le j \le \ell$ from $H'$.
As this operation removes $\ell$ vertices, but only $\ell+1$ edges, and $\ell \ge 1/\epsilon$,
the average degree does not decrease.
Finally, we construct $H''$ from $H'$ by simultaneously considering
all the maximal paths $v_0,v_1,\ldots,v_{\ell},v_{\ell+1}$, with all internal vertices of degree two,
and contracting each of such paths to a single edge $e'=v_0v_{\ell+1}$ and setting $w_{e'}=\bigcup_{0 \le j \le \ell} w_{v_jv_{j+1}}$. 
Observe that for each edge $e$ of $H''$ the size of $w_e$ is upper bounded by $2k^{2/\epsilon}(1/\epsilon+1)$,
as a contracted path consist of at most $\lfloor 1/\epsilon+1 \rfloor$ edges.

As $H''$ is of minimum degree at least $3$, we apply Corollary~\ref{cor:decomposition} 
to it, where $\gamma = 2k^{2/\epsilon}(1/\epsilon+1)$.
Let $H_0 = (V_0,E_0)$ and $\pathdecomp=(B_i)_{i=1}^q$ be
as defined in Corollary~\ref{cor:decomposition}.
Let $X \subseteq B_i^2$ be the set of all the vertices of $B_i^2$ corresponding to the edges of $E_0$,
including the vertices of $B_i^2$ that correspond to edges of $H'$ that were contracted into
some edge of $E_0$ (see Figure~\ref{fig2}).
As $|E_0| > |V_0|$ we have $|N_{G_i}(X)| < |X|$.
Clearly $X$ is of size at most $|E_0|(1/\epsilon+1) \le (4(\log_{3/2}|\cF|+2)+1)(1/\epsilon+1)$,
that is logarithmic in $|\cF|$, as $\epsilon$ is a constant.
It remains to show that we can lift $X$ to an improving set
of bounded pathwidth, while increasing its size only by a constant factor.

\begin{figure}
\begin{center}
\begin{tikzpicture}[scale=0.85]
  \tikzstyle{vertex}=[circle,fill=black,minimum size=0.20cm,inner sep=0pt]
  \tikzstyle{vertex2}=[circle,draw=black,fill=gray!50,minimum size=0.20cm,inner sep=0pt]
 	\tikzstyle{vertex3}=[circle,draw=black,fill=white,minimum size=0.20cm,inner sep=0pt]
 
  \tikzstyle{terminal}=[rectangle,draw=black,fill=white,minimum size=0.2cm,inner sep=0pt]

\node[vertex2] (a0) at (0,2){};
\node[vertex2] (a1) at (1,2){};
\node[vertex] (a2) at (2,2){};
\node[vertex] (a3) at (3,2){};
\node[vertex] (a4) at (4,2){};
\node[vertex] (a5) at (5,2){};
\node[vertex] (a6) at (6,2){};

\node[vertex2] (b0) at (0.5,0){};
\node[vertex2] (b1) at (1.5,0){};
\node[vertex] (b2) at (2.5,0){};
\node[vertex] (b3) at (3.5,0){};
\node[vertex] (b4) at (4.5,0){};
\node[vertex] (b5) at (5.5,0){};

\draw (a5) -- (b4) -- (a6) -- (b5) -- (a5);
\draw[dashed] (b4) -- (a4) -- (b3) -- (a3) -- (b2);
\draw (b2)-- (a2) -- (b4);

\draw (a0) -- (b0);
\draw (a1) -- (b1);
\draw (b0) -- (a2) -- (b1) -- (a4);	

\draw[snake=brace] (1.8,2.2) -- (6.2,2.2);
\draw (4,2.6) node {$Y_i=X$};

\draw[snake=brace] (-0.2,2.2) -- (1.2,2.2);
\draw (0.5,2.6) node {$Y_{i-1} \setminus Y_i$};

\draw[snake=brace] (5.7,-0.7) -- (2.3,-0.7);
\draw (4,-1.2) node {$N_{G_{i}}(Y_i)$};

\draw (2.5,-0.4) node {$a$};
\draw (4.5,-0.4) node {$b$};
\draw (5.5,-0.4) node {$c$};

\node[vertex] (a) at (7,1){};
\node[vertex] (b) at (8,1){};
\node[vertex] (c) at (9,1){};
\draw (a) edge[bend left] (b);
\draw[dashed] (a) edge[bend right] (b);
\draw (b) edge[bend left] (c);
\draw (b) edge[bend right] (c);
\draw (7,0.6) node {$a$};
\draw (8,0.6) node {$b$};
\draw (9,0.6) node {$c$};
\draw (8,1.5) node {$H_0$};

\end{tikzpicture}
\end{center}
\caption{The right graph is $H_0=(V_0,E_0)$ provided by Corollary~\ref{cor:decomposition}.
  The left side depicts the set $X$ corresponding to $E_0$, as well as lifting the set $Y_i=X$ to $Y_{i-1}$. Gray vertices belong to $G_{i-1}$ but not to $G_{i}$. The dashed path on the left between $a$ and $b$ in $H'$ is contracted into an edge of $H''$ on the right.}
\label{fig2}
\end{figure}
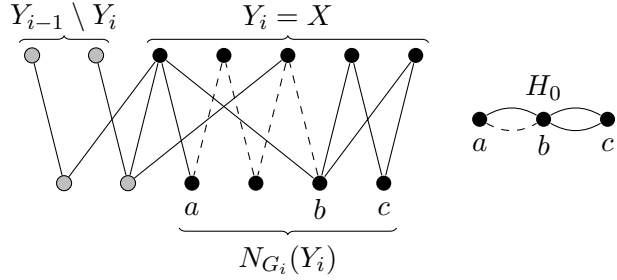

Let $Y_i = X$. For $j=i-1,\ldots,0$ set $Y_j = Y_{j+1} \cup (N_{G_j}(N_{G_j}(Y_j)) \cap B_j^1)$ (see Figure~\ref{fig2}).
Observe that at each step the size of $Y_j$ increases by a factor of at most $k+1$,
hence $|Y_0| \le |Y_i|(k+1)^i$ and moreover $Y_0$ is an improving set w.r.t. $\cF_0$.
Since $Y_0$ is of size logarithmic in $|\cF|$ it remains to show that $N_{\cg}[Y_0]$
is of constant pathwidth.

Create a sequence of subsets $\pathdecomp'=(B_i')_{i=1}^q$, by taking as $B_i'$
the set $(\bigcup_{e=uv \in E_0, u,v \in B_i} w_e \cap N_{\cg}[Y_0])$.
The size of each $B_i'$ is at most $(w+1)^2\gamma$, where $w$
is the width of $\pathdecomp$, hence it remains to show that $\pathdecomp'$
is indeed a path decomposition.
Each vertex of $N_{\cg}[Y_0]$ is within distance at most $2/\epsilon$
from some vertex of $X$, hence each vertex of $N_{\cg}[Y_0]$ 
is contained in some set $w_e$ for $e \in E_0$.
Similarly each edge of $\cg[N_{\cg}[Y_0]]$ is within distance
at most $2/\epsilon$ from some vertex of $X$, so it has both its endpoints
in some set $w_e$ for $e \in E_0$.
Since $\pathdecomp$ is a path decomposition each edge $e \in E_0$
has both its endpoints in some bag~$B_i$, therefore
$\bigcup_{1 \le i \le q} B_i' = N_{\cg}[Y_0]$
and each edge of $N_{\cg}[Y_0]$ has both its endpoints in some bag $B_i'$.
 Property (d) of Corollary \ref{cor:decomposition} implies
 that each $w_e$ contributes to $B_i'$ for values of $i$ forming
 an interval $I_e$. Moreover if for two edges $e_1, e_2 \in E_0$
 the intersection $w_{e_1} \cap w_{e_2}$ is non-empty, then
 by property (c) of Corollary~\ref{cor:decomposition} we know
 that the intervals $I_{e_1}$ and $I_{e_2}$ have non-empty intersection.
 This ensures that each vertex $v$ of $N_{\cg}[Y_0]$ appears in a set of bags $B_i'$
 forming an interval in the sequence $\pathdecomp'$,
 as each pair of intervals in $\{I_{e} : v \in w_e\}$ has non-empty intersection.

Therefore $Y_0$ is an improving set of logarithmic size and 
of constant pathwidth, which is a contradiction. Consequently
$|B_i^2| \le (1+\epsilon) |A_i|$, which finishes 
the proof of Claim~\ref{claim:new}.
\end{proof}

Lemma~\ref{lem:analysis} together with the algorithm
for searching improving sets of bounded pathwidth from Theorem~\ref{thm:search}
gives a polynomial time $(k+1+\epsilon)/3$-approximation algorithm
for \kSP for any constant $k$, proving Theorem~\ref{thm:main}.
In particular there is a $(4/3+\epsilon)$-approximation for 
the \tDM problem.

\section{Local search hardness}
\label{sec:hardness}

In this section we are going to show, that there is no 
algorithm verifying for a given $\cF_0 \subseteq \cF$, whether there exists 
an improving set (see Definition~\ref{def:improving}) of size at most $r$
in $f(r)\poly(|\cF|)$ time, even when $k=3$.
In fact we show a stronger hardness result, 
ruling out existence of an algorithm, that either 
finds a bigger disjoint family $\cF_1$ (without any restriction
on its distance from $\cF_0$), or verifies that there is no
improving set of size at most $r$.
That is exactly the notion of {\em permissive} parameterized local search
introduced by Marx and Schlotter in~\cite{permissive-ls} 
(for more information about parameterized local search see~\cite{fellows-ls,gaspers-ls,ls-survey}).

In our reduction, we use a standard W[1]-hard problem~\cite{mc-hardness}, namely \MC 
parameterized by the clique size.

\defproblem{\MC}{An undirected graph $G=(V,E)$, a positive integer $k$, and a color function $c : V \to \{0,\ldots,k-1\}$.}
{Does the graph $G$ contain a clique of size $k$, where each vertex is of different color?}

\begin{theorem}
\label{thm:w1-hardness}
There is a constant $\alpha > 0$, such that
given an instance $I=(G,k,c)$ of \MC one can in polynomial time construct
an instance $\cF \subseteq 2^U$ of \tSP, together with a disjoint subfamily $\cF_0 \subseteq \cF$ of size $|U|/3-1$, such that:
\begin{itemize}
  \item If $I$ is a YES-instance, then there exists a family $\cF_1 \subseteq \cF$
  of disjoint $|U|/3$ sets, such that $|\cF_0 \setminus \cF_1| + |\cF_1 \setminus \cF_0| \le \alpha k^2$,
  \item if there exists a disjoint subfamily $\cF_1 \subseteq \cF$ of size $|U|/3$,
  then $I$ is a YES-instance.
\end{itemize}
\end{theorem}

\begin{proof}
We start with a definition of a simple gadget, that will be used a couple of times
in the construction.
%TODO: figure in the journal version
\begin{definition}
For a positive integer $h \ge 1$ and a symbol $x$ an $(x,h)$-amplifier is a family $\cF_x \subseteq 2^{U_x}$ of sets of size~$3$, where 
\begin{align*}
U_x & = \{x_1,\ldots,x_{2\cdot 4^h-1}\}, \textrm{ and} \\
\cF_x & = \{\{x_i,x_{2i},x_{2i+1}\} : 1 \le i < 4^h\}
\end{align*}
\end{definition}

Let $I=(G=(V,E),k,c)$ be an instance of \MC.
W.l.o.g. we may assume that $k=4^h$, where $h$ is a positive integer,
since otherwise we may add universal vertices (adjacent to all other vertices).
We start with constructing an $(x,h)$-amplifier, which will
be called the {\em top amplifier},
and $(v,h)$-amplifier for each $v \in V$, called {\em vertex amplifiers}.
As the universe $U$ we take
\begin{align*}
U = & U_x \cup (\bigcup_{v \in V} U_v) \cup \{v_1', v_1'' : v \in V\} \cup \{s_{(i,j)} : 0 \le i < j < k\} \cup \{\ell_{i} : 1 \le i \le 2k\}\,.
\end{align*}
To the family $\cF$ we add all the sets of $\cF_x$ and $\cF_v$ for $v \in V$, as well as:
\begin{itemize}
  \item[(i)] sets $\{v_1, v_1', v_1''\}$ for $v \in V$,
  \item[(ii)] sets $\{x_{k+i},v_1',v_1''\}$ for $0 \le i < k$ for $v \in c^{-1}(i)$,
  \item[(iii)] sets $\{u_{k+c(v)}, v_{k+c(u)}, s_{(c(u),c(v))}\}$ for $uv \in E$, $c(u) < c(v)$,
  \item[(iv)] sets $\{v_{k+c(v)}, \ell_{2c(v)-1}, \ell_{2c(v)}\}$ for $v \in V$,
  \item[(v)] sets $\{\ell_{3i-2}, \ell_{3i-1}, \ell_{3i}\}$ for $1 \le i \le \lfloor 2k/3 \rfloor$ (note that $2k = 2 \cdot 4^h \equiv 2 \pmod 3$),
  \item[(vi)] consider all the elements $s_{(i,j)}$ in lexicographic order of pairs $(i,j)$,
  take subsequent triples of elements and add them to the family $\cF$, that is
  add sets 
  \begin{align*}
  \{&s_{(0,1)}, s_{(0,2)}, s_{(0,3)}\},\ldots, \{s_{(k-3,k-2)}, s_{(k-3,k-1)}, s_{(k-2,k-1)}\}
  \end{align*}
  (note that $\binom{k}{2} \equiv 0 \pmod 3$, since $(k-1) \equiv 0 \pmod 3$).
\end{itemize}

To finish the construction we create a disjoint family $\cF_0$ of size $|U|/3-1$ as follows:
\begin{itemize}
  \item add to $\cF_0$ sets $\{x_i, x_{2i}, x_{2i+1}\} \in \cF_x$ for $1 \le i < k$ such that $\lfloor \log_2 i \rfloor$ is odd.
  \item add to $\cF_0$ sets $\{v_i, v_{2i}, v_{2i+1}\} \in \cF_v$ for $v \in V$ and $1 \le i < k$, such that $\lfloor \log_2 i \rfloor$ is odd.
  \item add to $\cF_0$ all the sets from points (i), (v), (vi) of the construction of $\cF$.
\end{itemize}

Note that the size of $\cF_0$ equals $|U|/3-1$, as the only elements which are not covered by $\cF_0$ are $x_1$, $\ell_{2k-1}$ and $\ell_{2k}$.

\begin{claim}
\label{claim:1}
If $I$ is a YES-instance, then there exists a disjoint family $\cF_1 \subseteq \cF$ of size $|U|/3$,
such that $|\cF_1 \setminus \cF_0| + |\cF_0 \setminus \cF_1|= \Oh(k^2)$.
\end{claim}

\begin{proof}
Let $K \subseteq V$ be a solution to $I$, that is a multicolored clique of size $k$.
Construct a disjoint family $\cF_1$ as follows:
\begin{itemize}
\item[(a)] add to $\cF_1$ sets $\{x_i, x_{2i}, x_{2i+1}\} \in \cF_x$ for each $1 \le i < k$, such that $\lfloor \log_2 i \rfloor$ is even,
\item[(b)] add to $\cF_1$ sets $\{v_i, v_{2i}, v_{2i+1}\} \in \cF_x$ for $v \in K$ and $1 \le i < k$, such that $\lfloor \log_2 i \rfloor$ is even,
\item[(c)] add to $\cF_1$ sets $\{v_i, v_{2i}, v_{2i+1}\} \in \cF_x$ for $v \in V \setminus K$ and $1 \le i < k$, such that $\lfloor \log_2 i \rfloor$ is odd,
\item[(d)] for $0 \le i < k$ add to $\cF_1$ the set $\{x_{k+i}, v_1', v_1''\}$, where $v$ is the unique vertex of $K$ of color $i$,
\item[(e)] add to $\cF_1$ sets $\{v_1, v_1', v_1''\}$ for $v \in V \setminus K$,
\item[(f)] add to $\cF_1$ sets $\{u_{k+c(u)}, v_{k+c(v)}, s_{c(u),c(v)}\}$ for $u,v \in K$, $c(u) < c(v)$,
\item[(g)] add to $\cF_1$ sets $\{v_{k+c(v)}, \ell_{2c(v)-1}, \ell_{2c(v)}\}$ for $v \in K$.
\end{itemize}
A direct check shows that the above family is disjoint and covers all the elements of $U$, hence $|\cF_1| = |U|/3$.
Note that in the above construction of $\cF_1$ in each of the points (a), (d), (g) we add to $\cF_1$ only $\Oh(k)$ sets,
while in points (b), (f) we add to $\cF_1$ $\Oh(k^2)$ sets, whereas
in points (c) and (e) we add to $\cF_1$ sets that are present in $\cF_0$. 
Therefore the number of sets of $\cF_1$ which are not present in $\cF_0$ is upper bounded by a linear function in $k^2$.
\end{proof}

\begin{claim}
\label{claim:2}
If there exists a disjoint family $\cF_1$ of size $|U|/3$, then
$I$ is a YES-instance.
\end{claim}

\begin{proof}
Let $\cF_1 \subseteq \cF$ be any disjoint family of size $|U|/3$.
Since the element $x_1$ can be covered only by the set $\{x_1,x_2,x_3\}$, the family $\cF_1$ contains
all the sets $\{x_i, x_{2i}, x_{2i+1}\} \in \cF_x$ for $1 \le i < k$, where $\lfloor \log_2 i \rfloor$ is even,
%TODO: figure in the journal version
and consequently elements $x_{k+i}$ for $0 \le i < k$ are not covered by sets of $\cF_x$.
Therefore elements $x_{k+i}$ are covered by sets
from point (ii) of the construction of $\cF$,
hence for each $0 \le i < k$ in $\cF_1$ there is exactly one set $\{v_1,v_2,v_3\} \in \cF_1$ for $v \in c^{-1}(i)$,
and let $K$ be the set of those $k$ multicolored vertices.

We want to show that $K$ is a clique.
As for each $v \in K$ we have $\{v_1,v_2,v_3\} \in \cF_1$, 
the family $\cF_1$ contains all the sets $\{v_i,v_{2i},v_{2i+1}\}$
for $1 \le i < k$ where $\lfloor \log_2 i \rfloor$ is even.
Consequently elements $v_{k+i}$ for $0 \le i < k$, $i \neq c(v)$
are covered by sets from point (iii) of the construction of $\cF$.
Consider any pair $0 \le i < j < k$.
Denote as $u$ the unique vertex of $K \cap c^{-1}(i)$
and let $\{u_{k+j},v_{k+i},s_{(i,j)}\}$ be the set of $\cF_1$ covering $u_{k+j}$, where $v \in c^{-1}(j)$.
This implies that $v_{k+i}$ is not covered by a set of the $(v,h)$-amplifier,
hence $v_{1}$ is covered by the $(v,h)$-amplifier, i.e. by $\{v_1,v_2,v_3\}$.
Therefore $v \in K$ and the vertices of colors $i$ and $j$ of $K$
are adjacent.
Since $i$ and $j$ were selected arbitrarily, $K$ is a clique.
\end{proof}

The proof of Theorem~\ref{thm:w1-hardness} follows from Claim~\ref{claim:1} and Claim~\ref{claim:2}.
\end{proof}

Theorem~\ref{thm:w1-hardness}, together with the well-known fact that \MC
is W[1]-hard~\cite{mc-hardness} implies Theorem~\ref{thm:w1-intro}.

%\begin{corollary}
%There is no $f(r) |\cF|^{O(1)}$ time algorithm, that given a family $\cF \subseteq 2^U$
%of sets of size $3$ and its disjoint subfamily $\cF_0 \subseteq \cF$
%either finds a bigger disjoint family $\cF_1 \subseteq \cF$ or verifies
%that there is no disjoint family $\cF_1 \subseteq \cF$ such that $|\cF_0 \setminus \cF_1| + |\cF_1 \setminus \cF_0| \le r$, unless $FPT=W[1]$.
%\end{corollary}

\section{Future work and open problems}
\label{sec:conclusions}

One can try to continue the research direction
of Chan and Lau~\cite{lau}, who presented a
strengthening of the standard LP relaxation,
proving integrality gap of $(k+1)/2$
using a local search inspired analysis.
We would like to ask a question whether it is possible
to obtain some strengthened LP relaxation with integrality gap
$(k+c)/3$-for some constant~$c$.

Finally, we believe that it is worth looking into 
other problems, where local search algorithms were
applied successfully, such as {\sc $k$-Median}~\cite{k-median}
or {\sc Restricted Max-Min Fair Allocation}~\cite{svensson}.
A potential goal would be to design improved approximation local search algorithms
using non-constant size swaps in the spirit of the framework of this paper.

\section*{Acknowledgements}

We would like to thank Marcin Mucha for helpful discussions.

\bibliographystyle{abbrv}
\bibliography{3dm}

\end{document}